\documentclass[11pt]{article}

\input{setup.sty}

\title{Quantum Extremal Surfaces and the Holographic Entropy Cone}
 
\author[1]{Chris Akers,}
\emailAdd{cakers@mit.edu}
 
\author[2]{Sergio Hern\'andez-Cuenca}
\emailAdd{sergiohc@ucsb.edu}

\author[2]{and Pratik Rath}
\emailAdd{rath@ucsb.edu}
 
\affiliation[1]{Center for Theoretical Physics,\\
Massachusetts Institute of Technology, Cambridge, MA 02139, USA}

\affiliation[2]{Department of Physics, University of California, Santa Barbara, CA 93106, USA}

\abstract{Quantum states with geometric duals are known to satisfy a stricter set of entropy inequalities than those obeyed by general quantum systems.
The set of allowed entropies derived using the Ryu-Takayanagi (RT) formula defines the Holographic Entropy Cone (HEC).
These inequalities are no longer satisfied once general quantum corrections are included by employing the Quantum Extremal Surface (QES) prescription.
Nevertheless, the structure of the QES formula allows for a controlled study of how quantum contributions from bulk entropies interplay with HEC inequalities.
In this paper, we initiate an exploration of this problem by relating bulk entropy constraints to boundary entropy inequalities.
In particular, we show that requiring the bulk entropies to satisfy the HEC implies that the boundary entropies also satisfy the HEC.
Further, we also show that requiring the bulk entropies to obey monogamy of mutual information (MMI) implies the boundary entropies also obey MMI.}

\begin{document}
 
% Title Page
 
%\input{editionlegend.tex}
\maketitle

\section{Introduction}

The holographic principle has served as a remarkable guide for research in quantum gravity~\cite{tHooft:1993dmi,Susskind:1994vu}.
Its most prominent realization, the AdS/CFT correspondence~\cite{Maldacena_1999,Witten1998}, has in recent years elucidated the underlying mechanism for bulk gravitational observables to be encoded holographically on the boundary theory~\cite{Dong:2016eik,Faulkner:2017vdd, Cotler:2017erl,Hayden:2018khn}.
The derivation of these results has heavily relied on a progressively sophisticated understanding of holographic entanglement entropy.

The program of relating the von Neumann entropy of boundary regions to bulk quantities started with the groundwork of Ref.~\cite{Ryu:2006bv}, which provided the first version of this entry in the holographic dictionary, known as the Ryu-Takayanagi (RT) formula.
According to it, the von Neumann entropy $S(R)$ of a boundary region $R$ is given by the area $\mathcal{A}(\gamma_R)$ of a minimal-area bulk codimension-$2$ surface $\gamma_R$ homologous to $R$ relative to $\partial R$.
More explicitly,
\begin{equation}
\label{eq:RT}
    \SRT(R) = \min_{\gamma_R} \; \frac{\mathcal{A}(\gamma_R)}{4 G},
\end{equation}
where the minimization is over surfaces $\gamma_R$ contained in a time-symmetric slice and anchored to $R$, i.e., subject to the condition $\partial \gamma_R = \partial R$. 
This relation underpins the idea of geometry emerging from entanglement, an important theme in our understanding of holography~\cite{VanRaamsdonk:2010pw}.
The RT formula, \Eqref{eq:RT}, and its time-dependent~\cite{Hubeny:2007xt} and quantum generalizations~\cite{Faulkner:2013ana,Engelhardt:2014gca} have been exploited in various ways in the context of bulk reconstruction from boundary data~\cite{Dong:2016eik,Faulkner:2017vdd, Cotler:2017erl,Hayden:2018khn,Czech:2015qta,Czech:2016xec,Bao:2019bib,Bao:2020abm}. 

Given the insights that holography provides, it has been of great interest to study necessary and sufficient conditions for field theories to have holographic duals~\cite{Heemskerk:2009pn}.
Within such boundary theories, there is a ``code subspace'' of states whose bulk duals admit semiclassical geometric descriptions~\cite{Dong:2016eik, Harlow:2016vwg}. 
These so-called holographic states are a special sub-class of all quantum states, satisfying certain additional constraints.
The RT formula has allowed us to identify some of these constraints as inequalities that the entropies of boundary subregions need to satisfy. The first such genuinely holographic inequality that was discovered was the monogamy of mutual information (MMI)~\cite{Hayden:2011ag,Wall:2012uf}, an inequality on $3$ parties which reads
\begin{equation}\label{eq:mmi}
    \S{AB}+\S{AC}+\S{BC}\ge \S{A}+\S{B}+\S{C}+\S{ABC}.
\end{equation}
More generally, for any set of parties $[n]\equiv\{1,2,\dots,n\}$, one expects the RT formula to imply inequalities involving various combinations of the entropies of the $2^n-1$ subsystems associated to the nonempty subsets $\varnothing\ne I\subseteq[n]$.
A systematic study of such entropy inequalities obeyed by the RT formula was initiated by Ref.~\cite{Bao:2015bfa}.
The starting point is to, for each $n\ge1$, arrange the $2^n-1$ subsystem entropies of general mixed states on $n$ parties into entropy vectors $S\in\mathbb{R}^{2^n-1}$.
This allows one to think of entropy inequalities as bounding a space of allowed entropy vectors in $\mathbb{R}^{2^n-1}$. Ref.~\cite{Bao:2015bfa} was able to prove that this space is a polyhedral cone, i.e., a convex space bounded by a finite set of homogeneous, linear inequalities at each $n$, which they termed the Holographic Entropy Cone (HEC).
Since then, there has been a lot of interest in understanding the HEC~\cite{Hubeny:2018trv,Hubeny:2018ijt,Cuenca:2019uzx,He:2019ttu,He:2020xuo,Avis:2021xnz} and the general structure of entanglement in holographic states~\cite{Freedman:2016zud,Nezami:2016zni,Cui:2018dyq,Akers:2019gcv}.

An important caveat, however, is that the RT formula only works to leading order in $G$ when the amount of entropy in bulk matter is small, i.e., when the geometric $\mathcal{O}(1/G)$ term dominates.
The perturbative generalization of \Eqref{eq:RT} to all orders in $G$ utilizes the notion of a quantum extremal surface (QES)~\cite{Engelhardt:2014gca}.
A bulk surface $\gamma_R$ anchored to $R$ is called a QES if it extremizes the generalized entropy functional
\begin{equation}\label{eq:Sgen}
    \Sgen(R) = \frac{\mathcal{A}(\gamma_R)}{4 G} + \Sbulk(\sigma_{R}),
\end{equation}
where the last term is the von Neumann entropy of bulk quantum fields on any achronal homology region $\sigma_{R}$ defining the entanglement wedge of $R$ and satisfying $\partial \sigma_{R} = \gamma_R\cup R$.
Accordingly, the QES formula states that $S(R)$ is given by
\begin{equation}\label{eq:QES}
    \SQES(R) = \min \left\{ \underset{\gamma_R}{\ext} \; \Sgen(R) \right\},
\end{equation}
which instructs one to pick the minimal value over extrema of the generalized entropy if there happens to be more than one QES.
This generalization has been of paramount importance in recent computations of the unitary Page curve of black hole evaporation~\cite{Penington:2019npb,Almheiri:2019psf}, a remarkable development in our understanding of the black hole information problem~\cite{Penington:2019kki, Almheiri:2019hni,Almheiri:2019qdq}.\footnote{For a review of these topics, see Refs.~\cite{Almheiri:2020cfm,Raju:2020smc}.}

Importantly, in settings where \Eqref{eq:QES} receives contributions from the bulk entropy term $\Sbulk$ of \Eqref{eq:Sgen} in unrestricted ways, the arguments for the HEC inequalities no longer apply.
In fact, one expects that by choosing the bulk matter to have arbitrary entanglement structures, there are no entropy constraints on the boundary quantum state other than the universal ones obeyed by quantum states (such as strong subadditivity). Hence the interesting question to ask is what effect entropic constraints on the bulk matter fields have on the entanglement structure of the boundary state.

To make this problem more precise, let us introduce the notion of a QES entropy cone as the space of all entropy vectors compatible with the QES formula, \Eqref{eq:QES}, at a moment of time symmetry (analogous to the HEC for RT).\footnote{A generalization to nontrivial time-dependent settings will be no less subtle than that of understanding the analogue of the HEC for HRT, a problem which remains barely understood except for specific situations such as the low-dimensional case studied in Ref.~\cite{Czech:2019lps}.} As alluded to above, we expect such a QES cone to coincide with the general quantum cone of Ref.~\cite{pippenger2003inequalities}. In general, the QES cone can be shown to contain the quantum cone as a subset since one can construct situations where the area term can be neglected. Further, it seems plausible that the QES cone is the same as the quantum cone since one can use tensor networks with bulk legs to construct states where the QES formula applies~\cite{Hayden:2016cfa}.

One could then impose constraints on $\Sbulk$, beyond those of quantum mechanics, to obtain a constrained QES cone which will naturally be a subset of the original one. Such constrained QES cones cannot be arbitrarily small: regardless of how strong the constraints one imposes on the bulk entropy are, the QES formula will always be able to probe all possible entropies compatible with HEC inequalities by simply neglecting the bulk entropy term. Hence we find that constrained QES cones will always be supersets of the HEC and, in that sense, be quantum-corrected versions of the HEC. Indeed, we see that by varying the restrictions on the quantum contribution $\Sbulk$, constrained QES cones nicely interpolate between the ``classical'' HEC and the fully quantum entropy cone.

In this paper, we initiate the program of understanding quantum corrections to the HEC by looking at a couple of specific constrained QES cones. In particular, we investigate the QES cone 1) when $\Sbulk$ is constrained to obey all HEC inequalities and 2) when $\Sbulk$ is constrained to obey just MMI.
We refer to these as the HEC-constrained QES cone and MMI-constrained QES cone, respectively.

For instance, in order for the boundary to satisfy MMI, which is a $3$-party inequality, it was shown in Ref.~\cite{Akers:2019lzs} that the bulk state would have to satisfy a specific $7$-party inequality, namely,
\begin{equation}\label{eq:7mmi}
    \S{ABDG}+\S{ACEG}+\S{BCFG}\ge \S{A}+\S{B}+\S{C}+\S{ABCDEFG}.
\end{equation}
More generally, we will see that any given $n$-party holographic entropy inequality holds for $\SQES$ so long as the bulk state satisfies a certain $n'$-party inequality for $n'\gg n$, where $n'$ is generally doubly-exponential in $n$.
To make precise the relation between bulk and boundary entropy constraints, here we prove the following two important results:
\begin{itemize}
    \item \textbf{Result 1}: \quad Bulk HEC $\implies$ Boundary HEC
    \item \textbf{Result 2}: \quad Bulk MMI $\implies$ Boundary MMI
\end{itemize}

In more detail, \textbf{Result 1} states that the HEC-constrained QES cone is the same as the HEC.
We motivate this result in Sec.~\ref{sec:double} by using the idea of double holography, described in Ref.~\cite{Almheiri:2019hni,Almheiri:2019psy}.
Intuitively, a doubly holographic setup ensures that although the boundary entropies computed using the QES formula receive a large contribution from the bulk entropy, they can secretly be thought of as being computed by an area in a higher-dimensional bulk.
The standard arguments for proving HEC inequalities obeyed by the RT formula can then be generalized, elucidating a connection between the bulk HEC and the boundary HEC.

Taking inspiration from the lessons learned from double holography, we then prove \textbf{Result 1} formally in Sec.~\ref{sec:contraction} by employing the formalism of contraction maps pioneered by Ref.~\cite{Bao:2015bfa} (see also Ref.~\cite{Avis:2021xnz} for more details). 
After reviewing the formalism in Sec.~\ref{sec:contreview}, we prove \textbf{Result 1} in Sec.~\ref{sec:proofHEC}.
The proof involves an elegant reinterpretation of the original contraction maps used to prove boundary HEC inequalities in the context of the RT formula. Despite the way we motivate it, we emphasize that \textbf{Result 1} is independent from double holography and applies more generally to any holographic field theory.

In Sec.~\ref{sec:MMI}, we then initiate a more controlled analysis of the implications of bulk entropy constraints on the boundary state. In particular, we consider the situation in which we constrain the bulk entropy to obey MMI for arbitrary subregions, which defines the MMI-constrained QES cone.
Remarkably, we find that this suffices for the boundary state to also satisfy MMI for arbitrary subregions, thereby proving \textbf{Result 2}.
In other words, the aforementioned $7$-party inequality, \Eqref{eq:7mmi}, is not only a HEC inequality (as implied by \textbf{Result 1}), but moreover it is weak in the sense that it is implied already by the $3$-party MMI inequality.
Thus, we see that the MMI-constrained QES cone is at least as small as the cone obtained by imposing all instances of MMI on quantum entropies, i.e., the MMI cone.

Finally, we conclude in Sec.~\ref{sec:disc} with a summary of our results and a discussion of future directions of the research program we have initiated here.

\section{Motivation: Double Holography} \label{sec:double}

In this section we motivate our first result, namely, the fact that the bulk HEC implies the boundary HEC.
To do so, we consider a doubly holographic setup where the $d$-dimensional gravity dual to a holographic quantum system itself has an effective holographic description in a $(d+1)$-dimensional theory of gravity.
This setup has been exploited to study questions involving strong quantum effects\footnote{Most notably, in the derivation of the unitary Page curve of black hole evaporation.} in $d$-dimensional quantum gravity, which can be comfortably analyzed by working with classical gravity on a $(d+1)$-dimensional bulk~\cite{Emparan:2006ni,Myers:2013lva,Almheiri:2019hni,Rozali:2019day,Chen:2020uac,Chen:2020hmv,Hernandez:2020nem,Chandrasekaran:2020qtn,Bousso:2020kmy,Neuenfeld:2021bsb,Geng:2020qvw}.

Consider a $d$-dimensional holographic boundary conformal field theory (BCFT).
The bulk dual to such a theory was described in Refs.~\cite{Takayanagi:2011zk,Fujita:2011fp} and involves a $(d+1)$-dimensional bulk with an end-of-the-world (EOW) brane anchored to the boundaries of the BCFT.
In fact, such a holographic theory involves two layers of holography and thus admits essentially three descriptions, as seen in Fig.~\ref{fig:doublehol}:
\begin{itemize}
    \item {\bf Description 1}: A UV complete, $d$-dimensional BCFT where the CFT couples to $(d-1)$-dimensional boundary defects which admit a holographic description.
    \item {\bf Description 2}: An effective $d$-dimensional theory where the CFT now couples to a $d$-dimensional gravitating brane replacing the boundary defects of Description 1.
    \item {\bf Description 3}: A $(d+1)$-dimensional theory of gravity with an EOW brane.
\end{itemize}

\begin{figure}[ht!]
    \centering
    \includegraphics[width=.8\textwidth]{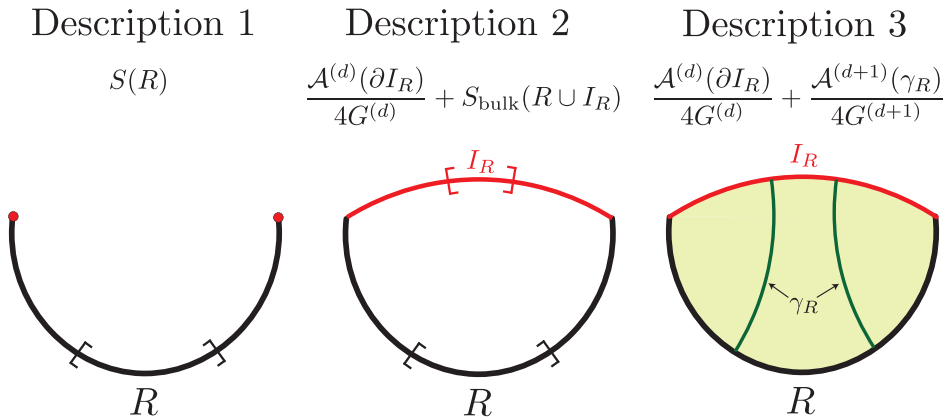}
    \caption{The three descriptions in the doubly holographic setup and their respective calculations for the entropy of subregion $R$. Description 1: A $d$-dimensional BCFT with a subregion $R$. Description 2: A $d$-dimensional CFT coupled to gravity on a $d$-dimensional brane (red); the entropy is computed using the QES formula including contributions from the island $I_R$. Description 3: A $(d+1)$-dimensional bulk with an EOW brane; the entropy is computed using the RT surface $\gamma_R$ (green) which is anchored to subregion $R$ and its island $I_R$. Note that we have assumed the RT surface to be connected in Description 3, corresponding to $R$ having a nontrivial island in Description 2.}
    \label{fig:doublehol}
\end{figure}

As described in Ref.~\cite{Almheiri:2019hni,Almheiri:2019psy}, the entropy of a subregion $R$ of the BCFT in Description 1 can be related to quantities in the other descriptions using the QES formula (see Fig.~\ref{fig:doublehol}),
\begin{subequations}
\label{eq:doubleQES}
\begin{align}
\label{eq:doubleQES2}
    \SQES(R) &= \min \left\{\underset{I_R}{\ext} \; \frac{\mathcal{A}^{(d)}(\partial I_R)}{4 G^{(d)}} + \Sbulk(R \cup I_R) \right\}\\
\label{eq:doubleQES3}
    &\approx \min \left\{ \underset{\gamma_R}{\ext} \; \frac{\mathcal{A}^{(d)}(\partial I_R)}{4 G^{(d)}} + \frac{\mathcal{A}^{(d+1)}(\gamma_R)}{4 G^{(d+1)}} \right\}.
\end{align}
\end{subequations}
Here, the first line is a computation in Description 2 using the QES formula and $I_R$ can be thought of as an island of $R$, living on the $d$-dimensional brane. Equivalently, the entropy of $R$ could be computed using the quantum maximin procedure which gives the same answer as \Eqref{eq:doubleQES2}~\cite{Akers:2019lzs}.
The second line is a computation in Description 3 using the simpler RT formula (i.e. neglecting any quantum corrections), where the bulk entropy term $\Sbulk(R \cup I_R)$ can be calculated using the $(d+1)$-dimensional area of a surface $\gamma_R$ homologous to $R \cup I_R$.

By construction, it is clear that all entropies computed in Description 2, including those for regions on the brane, obey all HEC inequalities. This is because this ``bulk entropy" can be effectively computed in Description 3 by the RT formula, which satisfies the HEC as defined in Ref.~\cite{Bao:2015bfa}. Hence we say that the ``bulk HEC'' is obeyed.
Now, we would like to consider entropies in Description 1, i.e., those given by the full expression in \Eqref{eq:doubleQES2}, and show that these also lie within the HEC.
This should be understood as the ``boundary HEC'', which explicitly coexists with the bulk HEC in this setup through \Eqref{eq:doubleQES}.

\begin{figure}
    \centering
    \includegraphics[width=.5\textwidth]{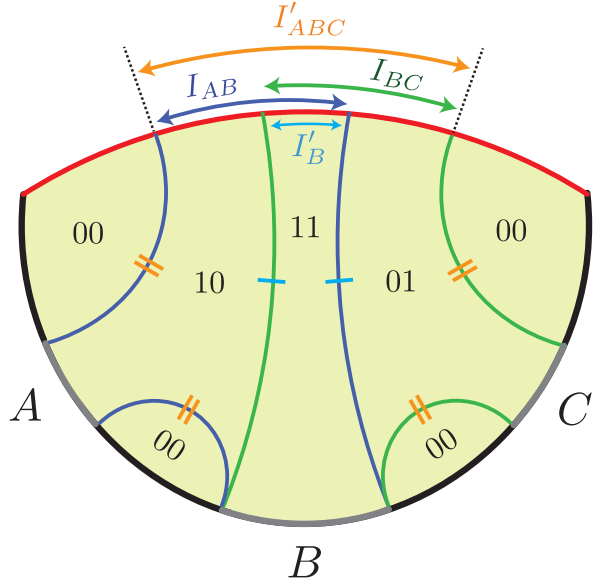}
    \caption{The entropies of subsystems $AB$ and $BC$ are computed by RT surfaces $\gamma_{AB}$ (blue) and $\gamma_{BC}$ (green), together with brane contributions coming from the boundaries of their associated islands on the brane, $I_{AB}$ and $I_{BC}$, respectively. The homology regions generating these RT surfaces correspond to unions of bulk regions, namely: $\gamma_{AB}$ is obtained from the boundary of $10$ and $11$, while $\gamma_{BC}$ comes from $11$ and $01$.
    Similarly, the islands whose boundaries contribute to the relevant entropies can be formed from analogous unions of brane regions which carry the same labels as their adjacent bulk regions.
    All these bulk and brane regions can be rearranged into homology regions and islands for the subsystems $B$ and $ABC$. In particular, $11$ alone gives a surface $\gamma_{B}'$ (single line) homologous to $B$ with island $I_{B}'$, and the union of $10$, $11$ and $01$ gives a surface $\gamma_{ABC}'$ (double line) homologous to $ABC$ with island $I_{ABC}'$. Thus the same bulk and brane labels allow one to keep track of homology regions and islands for left-hand and right-hand side subsystems in \Eqref{eq:SSA}.}
    \label{fig:braneSSA}
\end{figure}

Since the general idea will go through for all holographic inequalities in a similar fashion, we focus on the simple example of proving strong subadditivity (SSA) for subregions in Description 1.
Given subregions $A$, $B$ and $C$ in Description 1, we would like to show that
\begin{align}\label{eq:SSA}
        \S{AB}+\S{BC} \ge \S{B}+\S{ABC}.
\end{align}
A useful guide to the following calculation is provided by Fig.~\ref{fig:braneSSA}.
We can compute the left-hand side using \Eqref{eq:doubleQES} to obtain
\begin{align}\label{eq:SSAproof1}
    \S{AB}+\S{BC} &= \frac{\mathcal{A}^{(d)}(\partial I_{AB})}{4 G^{(d)}} + \Sbulk(AB \cup I_{AB}) +\frac{\mathcal{A}^{(d)}(\partial I_{BC})}{4 G^{(d)}} + \Sbulk(BC \cup I_{BC})\\
    &= \frac{\mathcal{A}^{(d)}(\partial I_{AB})}{4 G^{(d)}} + \frac{\mathcal{A}^{(d+1)}(\gamma_{AB})}{4 G^{(d+1)}} +\frac{\mathcal{A}^{(d)}(\partial I_{BC})}{4 G^{(d)}} + \frac{\mathcal{A}^{(d+1)}(\gamma_{BC})}{4 G^{(d+1)}},
\end{align}
where $\gamma_{AB}$ and $\gamma_{BC}$ are respectively RT surfaces anchored to $AB\cup I_{AB}$ and $BC\cup I_{BC}$.
These RT surfaces divide up the $(d+1)$-dimensional bulk and $d$-dimensional brane into four classes of regions.
Each region can be labelled by a bitstring which encodes inclusion / exclusion with respect to the homology regions and islands of $AB$ and $BC$ in the form of binary variables~\cite{Bao:2015bfa}. For instance, the region labelled by $10$ in Fig.~\ref{fig:braneSSA} receives this bitstring label because it is included in the homology region of $AB$ (bounded by blue lines) but excluded from that of $BC$ (bounded by green lines).\footnote{This notation will be described in more detail in Sec.~\ref{sec:contreview}, as it is crucial for a general discussion of the arguments leading to the holographic entropy inequalities.}
The labelling of $(d+1)$-dimensional bulk regions in Description 3 is associated to a corresponding labelling of $d$-dimensional brane regions in Description 2 -- in Fig.~\ref{fig:braneSSA}, brane regions carry the same labels as their adjacent bulk regions.

One can now reproduce the cut-and-paste procedure prescribed in Ref.~\cite{Bao:2015bfa}, now for both RT and island contributions, in terms of these labelled bulk and brane regions -- see Fig.~\ref{fig:braneSSA} for more details. In particular, the RT surfaces and island boundaries can be rearranged into a different set of surfaces, $\gamma_{B}'$ and $\gamma_{ABC}'$, and island boundaries, $\partial I_{B}'$ and $\partial I_{ABC}'$, now homologically associated to subsystems $B$ and $ABC$, respectively.
By the minimization in the QES formula, we are guaranteed that these will obey
\begin{align}\label{eq:SSAproof2}
    \S{AB}+\S{BC} &=  \frac{\mathcal{A}^{(d)}(\partial I_{B}')}{4 G^{(d)}} + \frac{\mathcal{A}^{(d+1)}(\gamma_{B}')}{4 G^{(d+1)}} +\frac{\mathcal{A}^{(d)}(\partial I_{ABC}')}{4 G^{(d)}} + \frac{\mathcal{A}^{(d+1)}(\gamma_{ABC}')}{4 G^{(d+1)}}\\
    &\geq \S{B}+\S{ABC},
\end{align}
thereby proving SSA.

Thus, by virtue of double holography, we see that $d$-dimensional and $(d+1)$-dimensional contributions can be handled in a coordinated way using standard cut-and-paste arguments. Respectively, these contributions can be seen as being associated to Descriptions 2 and 3.
From the perspective of Description 3, the above argument can be understood as a simple generalization of the arguments in Ref.~\cite{Bao:2015bfa} with the inclusion of island boundary terms of the form $\frac{\mathcal{A}^{(d)}(\partial I)}{4 G^{(d)}}$ in the RT formula.
Alternatively, from the perspective of Description 2, it can be understood as a generalization where the islands associated to the QES surfaces include bulk entropy contributions of a holographic type, behaving as $(d+1)$-dimensional contributions growing off of the islands.
Since Description 3 will no longer be available when abstracting away from double holography, the latter perspective will be the more useful one in what follows.
Indeed, the basic strategy of handling bulk entropy contributions in terms of bitstring labels induced by QES surfaces will still be a powerful one and serve as the cornerstone for proving {\bf Result 1} in Sec.~\ref{sec:contraction}.

Although we considered a time-independent situation, where the cut-and-paste technique of Ref.~\cite{Bao:2015bfa} can be directly applied, SSA and MMI can be proved in time-dependent settings.
For these inequalities, one can use the maximin techniques of Refs.~\cite{Wall:2012uf,Akers:2019lzs}, but now with the inclusion of extra boundary terms.\footnote{Note that there are inequalities that are not provable using the techniques of maximin~\cite{Rota:2017ubr}.}

To summarize, we have exemplified how the usual constructs for holographic entropy inequalities can be applied to both bulk and boundary entropies in a coordinated way in the setting of double holography. Within this framework, it becomes clear that the bulk HEC holds (by double holography) and so does the boundary HEC (because island contributions do not spoil the usual cut-and-paste arguments). 

In fact, it is plausible that requiring the bulk HEC to hold for arbitrary subregions (including those on the brane), could imply the existence of a higher-dimensional geometric description where the entropies are computed using the RT formula.\footnote{In certain cases, uniqueness of such a geometry would be guaranteed by the results of Refs.~\cite{Bao:2019bib,Bao:2020abm}.}
This in turn would imply the boundary HEC by using the corresponding doubly holographic picture.
Nevertheless, the construction of such a holographic dual given the boundary entropies would be a nontrivial task, if possible at all.
Thus, in Sec.~\ref{sec:contraction}, we bypass the usage of double holography as an intermediate tool.
Instead, we elucidate a more direct mechanism by which the bulk HEC implies the boundary HEC in full generality.

\section{Holographic Entropy Inequalities} \label{sec:contraction}

In Sec.~\ref{sec:contreview}, we review the proof technique of contraction maps utilized in Ref.~\cite{Bao:2015bfa} to prove holographic entropy inequalities.
In particular, we describe this method without resorting to graph theory, in a geometric language that we believe should be more transparent for the AdS/CFT community.
In Sec.~\ref{sec:proofHEC}, we prove \textbf{Result 1}, i.e., the fact that the bulk HEC implies the boundary HEC, using a novel application of contraction maps.

\subsection{Proof-by-Contraction} \label{sec:contreview}

The \textit{proof-by-contraction} method, originally developed by Ref.~\cite{Bao:2015bfa}, is a combinatorial method for proving holographic entropy inequalities when using the RT formula. This technique has by now been extended in various directions (see Ref.~\cite{Avis:2021xnz} for an upgraded version and Ref.~\cite{Bao:2020zgx} for a generalization), but always formalized in the language of graph theory. Since graph models of holographic entanglement\footnote{See Ref.~\cite{Bao:2015bfa} for more details on these.} have not played any role in our discussion so far, we will refrain from introducing them at this point. Instead, we now proceed to present the proof-by-contraction technique in a general geometric setting, which we believe will provide an explanation of its inner workings that feels more natural to holographers.

This proof method is a generalization of a strategy first employed in Ref.~\cite{Headrick:2007km} for a holographic proof of SSA, and later applied in Ref.~\cite{Hayden:2011ag} to the proof of MMI. In these first appearances, the setting was in fact geometric, so it is natural to try to reformulate the proof-by-contraction method in geometric terms. Hence, it will be useful to reproduce these early arguments by carrying along the proof of MMI as an example.

Without loss of generality, our MMI example can be described on the geometric configuration shown in Fig.~\ref{fig:mmilhs}. We emphasize though that the proof-by-contraction method does not assume anything about the underlying geometry, other than that it be a Riemannian manifold (e.g. arising from a time-symmetric slice of a Lorentzian manifold where the RT formula applies) with standard or asymptotic codimension-$1$ boundaries.\footnote{Geometries may have arbitrarily complicated topology, be non-smooth, involve multiple connected components, have any number of standard manifold or asymptotic boundaries, or even any kind of asymptotics, not necessarily AdS.
Similarly, boundary regions may consist of multiple connected components as well, be adjoining or not, be compact or not, or cover entire boundaries. In a holographic context, it may be important to restrict to the boundary being a convex surface~\cite{Sanches:2016sxy,Nomura:2018kji}.}

\begin{figure}
    \centering
    \includegraphics[width=.5\textwidth]{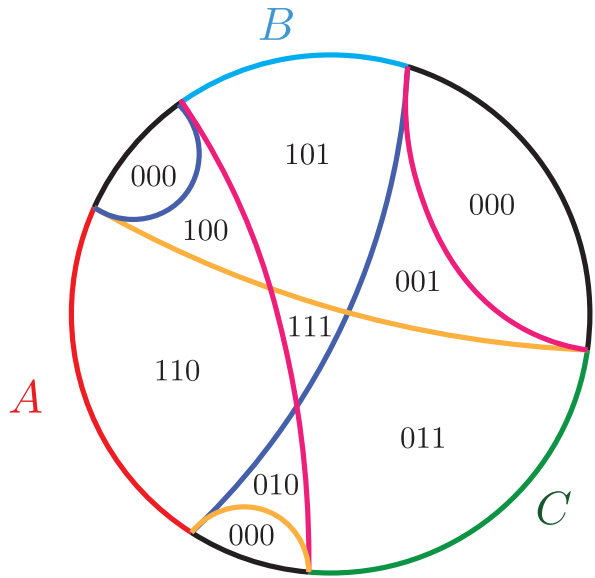}
    \caption{Example of a $3$-party holographic configuration illustrating the proof-by-contraction method for MMI. The relevant parties are $A$, $B$ and $C$, with boundary regions colored in red, light blue and green, respectively -- their complement or purifier is shown in black. For the three terms on the left-hand side of the MMI inequality, $AB$, $AC$ and $BC$, the RT surfaces are represented by the interior curves in blue, yellow and pink, respectively. The resulting homology regions partition the bulk geometry into subregions labelled by $2^3$ different bitstrings in $\{0,1\}^3$. These bitstrings are specified by inclusion/exclusion as explained in the main text.}
    \label{fig:mmilhs}
\end{figure}

A general candidate inequality will be canonically written as 
\begin{equation}
\label{eq:ineq}
    \sum_{l=1}^L \alpha_l S_{I_l} \ge \sum_{r=1}^R \beta_r S_{J_r},
\end{equation}
such that the coefficients $\alpha_l$ and $\beta_r$ are all positive, and $L$ and $R$ are the number of entropy terms appearing on each side. For an inequality written as in \Eqref{eq:ineq} to hold, it must be the case that its left-hand side (LHS) is guaranteed to be no smaller than its right-hand side (RHS) in any imaginable configuration. The basic idea of the proof-by-contraction is to devise an organizing principle that allows one to compare LHS terms to RHS ones in full generality. This turns out to be amenable to a combinatorial formulation which can be set up step-wise as follows:
\begin{enumerate}
    \item\label{s1} Collect the RT surfaces and their corresponding homology regions for all $L$ subsystems appearing on the LHS of the given inequality.
    \item\label{s2} Partition the bulk into the $2^L$ subregions allowed by inclusion/exclusion inside each of the $L$ homology regions (some subregions may be empty).
    \item\label{s3} Use the inclusion/exclusion subregions to reconstruct the LHS homology regions, and utilize them suitably to construct candidate homology regions for RHS subsystems -- the latter will generically not be associated to minimal RT surfaces.
    \item\label{s4} Express areas of surfaces bounding homology regions in terms of the subregions contributing to them. For the LHS, these will yield back entropy areas, whereas for the RHS this will yield some possibly nonminimal areas.
    \item\label{s5} Devise a diagnostic to compare the LHS entropy areas to the resulting RHS areas associated to the so-constructed homology regions.
\end{enumerate}
We now describe each of these steps in detail, referring to the MMI example when helpful.

Step \ref{s1} requires us to compute the RT surfaces associated to all subsystem entropies on the LHS of the inequality. For instance, for MMI these are the interior blue, yellow and pink curves in Fig.~\ref{fig:mmilhs}. If $\gamma_I$ is the RT surface of the region $R_I$, let $\sigma_I$ be its corresponding homology region satisfying $\partial\sigma_I = \gamma_I \cup R_I$. We will take these homology regions $\sigma_I$ to be open sets. Each LHS subsystem $I_l$ has an associated $\sigma_{I_l}$ subregion.

For step \ref{s2}, we will use these $\sigma_{I_l}$ to discretize the geometry as follows. Any bulk point $p\in M$ can be classified by whether it is inside or outside of $\sigma_I$ for every LHS subsystem $I$, a binary choice that has to be made $L$ times, one for each such term. Length-$L$ bitstrings $x\in\{0,1\}^L$ can thus be used to encode all regions in $M$ where points can be located in terms of LHS homology regions -- see Fig.~\ref{fig:mmilhs}. In particular, these subsets obtained by inclusion/exclusion in homology regions can be written as
\begin{equation}
\label{eq:sigmaxs}
    \sigma(x) = \bigcap_{l=1}^L \sigma_{I_l}^{x_l}, \qquad \sigma_I^b =
    \begin{cases}
        \sigma_I \quad &\text{if} \quad b=1, \\
        \sigma_I^\complement \quad &\text{if} \quad b=0. \\
    \end{cases}
\end{equation}
These $\sigma(x)$ sets partition $M$ into $2^L$ disjoint subregions, some of which may be empty or consist of multiple connected components as shown in Fig.~\ref{fig:mmilhs}.

Step \ref{s3} employs the subregions in \Eqref{eq:sigmaxs} to build back the LHS homology regions. Namely, since by construction $\sigma(x)$ is a nonempty subset of $\sigma_{I_l}$ if and only if $x_l=1$,
\begin{equation}
\label{eq:lhshomo}
    \sigma_{I_l} = \bigcup_{x:x_l=1} \sigma(x).
\end{equation}
For example, the homology region for $AB$ in Fig.~\ref{fig:mmilhs} consists of subregions labelled by $110$, $100$, $111$ and $101$, all with $x_A=1$, as required.

Crucially, these $\sigma(x)$ subregions can also be utilized to construct homology regions for the RHS subsystems. For a region $\sigma_{J_r} \subseteq M$ made up from $\sigma(x)$ subregions to be a valid homology region for a boundary region $R_{J_r}$, there are certain $\sigma(x)$ which one is obliged to include/exclude. In particular, it must be the case that $\sigma(x)\subseteq\sigma_{J_r}$ whenever $\sigma(x)$ is adjacent to $R_{J_r}$, and $\sigma(x) \not \subseteq\sigma_{J_r}$ whenever $\sigma(x)$ is only adjacent to other boundary regions $R_{J_r'}\ne R_{J_r}$. Once these constraints are satisfied, one is free to include any other $\sigma(x)$ in $\sigma_{J_r}$. 

To make these conditions more explicit, it will be convenient to think of general pure states on a set of parties $[n+1]$ where $n+1$ labels the complement to all $n$ parties.\footnote{Recall the notation $[n+1]\equiv \{1,2,\dots,n+1\}$, and the standard identification $1\leftrightarrow A$, $2\leftrightarrow B$, \dots, where for the purifier we reserve $n+1\leftrightarrow O$.}
The desired constraint for valid homology can then be phrased in terms of so-called \textit{occurrence bitstring}. 
Such a bitstring $x^{(i)}$ for a party $i\in[n+1]$ is defined bit-wise by a binary Boolean function that determines the set of LHS subsystems $I_l$ in which party $i$ occurs, i.e.,
\begin{equation}
\label{eq:occur}
    x^{(i)}_l = \delta(i\in I_l).
\end{equation}
In words, $x^{(i)}$ is the label specifying the bulk region $\sigma(x^{(i)})$ adjacent to the boundary region labelled by $i$.
% Notice that this includes an occurrence bitstring for the purifying region $n+1$, which is always an all-zero tuple given that $n+1$ never appears in any subsystem. 
In the MMI example, we get $x^{A}=110$ (because $A$ occurs in $AB$ and $AC$, but not in $BC$), and we see that the bulk region $\sigma(x^{A})$ this bitstring labels is indeed the one adjacent to boundary region $A$ in Fig.~\ref{fig:mmilhs}.
For each $i\in[n+1]$, these are the relevant bulk subregions involved in the homology constraint discussed above. Hence, this condition becomes the requirement that for occurrence bitstrings $\sigma(x^{(i)})\subseteq \sigma_{J_r}$ if and only if $i\in J_r$ -- for any other bitstring one is free to choose whether to include $\sigma(x)$ as part of $\sigma_{J_r}$ or not. These conditions and choices can all be captured by introducing a general map $f:\{0,1\}^L\to\{0,1\}^R$
specifying RHS homology regions via (cf. \Eqref{eq:lhshomo} for the LHS)
\begin{equation}
\label{eq:rhshomo}
    \sigma_{J_r} = \bigcup_{x:f(x)_r=1} \sigma(x),
\end{equation}
and subject to a constraint on all occurrence bitstrings for $i\in[n+1]$ of the form
\begin{equation}
\label{eq:occurcond}
    f(x^{(i)})_r = \delta(i\in J_r),
\end{equation}
guaranteeing that \Eqref{eq:rhshomo} is a valid homology region for $J_r$. The freedom in including other $\sigma(x)$ subsets in $\sigma_{J_r}$ corresponds to varying the map $f$ while preserving \Eqref{eq:occurcond}. As noted previously, the RHS homology region $\sigma_{J_r}$ resulting from a choice of map $f$ will generically not be bounded by a minimal surface for $J_r$ in the RT sense. For instance, applying \Eqref{eq:rhshomo} to the RHS subsystem $J_4=ABC$ in the MMI inequality, the constraint \Eqref{eq:occurcond} would instruct us to include $\sigma(x^{(i)})\subseteq\sigma_{ABC}$ for all $i\in[3]$, but not for $i=4$ (since $O$ does not appear in $ABC$). In Fig.~\ref{fig:mmilhs}, this means $\sigma_{ABC}$ must include regions labelled by $110$, $101$ and $011$, and exclude the ones labelled by $000$, which we see makes perfect sense as minimal requirements for any $ABC$ homology region built out of $\sigma(x)$ subregions. In addition, $f(x)_4$ would be free to take any value for any non-occurrence bitstring, e.g. $f(010)_4=1$ and $f(111)_4=0$ would respectively correspond to having $\sigma(010)\subseteq\sigma_{ABC}$ and $\sigma(111)\not\subseteq\sigma_{ABC}$.

Having understood how bulk subregions $\sigma(x)$ can be used to construct homology regions for both LHS and RHS subsystems, we can now proceed to step \ref{s4}. The crucial observation is that this partitioning of $M$ into $\sigma(x)$ subregions also induces a discretization of the RT surfaces which will allow for a cut-and-paste argument to compare areas of LHS and RHS surfaces. The subsets in \Eqref{eq:sigmaxs} can be used to chop RT surfaces into pieces by adjacency using the following object,
\begin{equation}
\label{eq:adjacent}
    \gamma(x,x') = \overline{\sigma(x)} \bigcap_{\text{codim-$1$}} \overline{\sigma(x')},
\end{equation}
where overlines denote closures and we have defined an intersection operator which yields the empty set if and only if the intersection of the two sets is of codimension higher than $1$. In other words, $\gamma(x,x')$ will be empty unless the subsets $\sigma(x)$ and $\sigma(x')$ are adjacent to each other along a codimension-$1$ surface, so $\gamma(x,x')$ will always be a subregion of an RT surface lying between bulk regions labelled by $x$ and $x'$.\footnote{At first sight, it could naively seem that two bitstrings $x$ and $x'$ will be labelling adjacent regions if and only if they differ by a single bit. This is in fact not true: bitstrings $x$ and $x'$ differing in multiple bits may be labelling adjacent regions if different subsystems share connected components of RT surfaces. In such cases, crossing the relevant RT surface would be associated to flipping more than one bit.}
Furthermore, we can identify which RT surface $\gamma(x,x')$ is part of, namely: $\gamma(x,x')$ is a nontrivial portion of $\gamma_{I_l}$ if and only if $x$ and $x'$ differ in the $l^{\text{th}}$ bit $x_l$ (i.e., they lie on opposite sides of $\gamma_{I_l}$) and $\gamma(x,x')\ne\varnothing$ (i.e., $\sigma(x)$ and $\sigma(x')$ are adjacent).

The power of this is that we can now reconstruct the RT surface of any LHS subsystem $I_l$ by piecing together appropriate portions $\gamma(x,x')$ by iterating over bulk subregions through
\begin{equation}
\label{eq:gammapiec}
    \gamma_{I_l} = \bigcup_{x,x':x_l \ne x_l'} \gamma(x,x').
\end{equation}
% where the iteration above is over all pairs of bitstrings obeying $x_l \ne x_l'$ condition. 
Furthermore, since no two $\gamma(x,x')$ have a codimension-$1$ intersection, \Eqref{eq:gammapiec} can also be used to compute entropies via the RT formula in \Eqref{eq:RT},
\begin{equation}
\label{eq:sibulkxx}
    \SRT(I_l) = \frac{1}{4G}\sum_{x,x'} \abs{x_l - x_l'} \mathcal{A}\left[\gamma(x,x')\right],
\end{equation}
where $\abs{x_l - x_l'}$ is a convenient indicator function implementing the condition that $x_l\ne x_l'$ algebraically. Notice that all we needed to obtain \Eqref{eq:sibulkxx} was a specification of the homology region of $I_l$ in terms of bulk subregions $\sigma(x)$. We have the same ingredients for the RHS subsystems from \Eqref{eq:rhshomo}, except the homology regions we can construct for RHS subsystems this way are not guaranteed to be bounded by surfaces of minimal area. This means that the homology regions we build for the RHS will yield
\begin{equation}
\label{eq:sjbulkxx}
    \SRT(J_r) \le \frac{1}{4G}\sum_{x,x'} \abs{f(x)_r - f(x)_r'} \mathcal{A}\left[\gamma(x,x')\right].
\end{equation}

We are now in a position to complete the argument with step \ref{s5} by comparing LHS and RHS terms in a candidate holographic inequality as written in \Eqref{eq:ineq}. It will be convenient to introduce the notion of a \textit{weighted Hamming distance} $d_w$ which, given a weight vector $w\in\mathbb{R}^m$, for any pair of bitstrings $y,y'\in\{0,1\}^m$, is defined as
\begin{equation}
    d_w(y,y') = \sum_{k=1}^m w_k \abs{y_k-y_k'}.
\end{equation}
With this notation in hand, and using \Eqref{eq:sibulkxx}, the LHS of \Eqref{eq:ineq} translates into
\begin{equation}
\label{eq:lhsend}
    \sum_{l=1}^L \alpha_l \SRT(I_l) = \frac{1}{4G}\sum_{x,x'} d_\alpha(x,x') \mathcal{A}\left[\gamma(x,x')\right],
\end{equation}
whereas using \Eqref{eq:sjbulkxx}, the RHS of \Eqref{eq:ineq} obeys
\begin{equation}
\label{eq:rhsend}
    \frac{1}{4G}\sum_{x,x'} d_\beta(f(x),f(x')) \mathcal{A}\left[\gamma(x,x')\right] \ge \sum_{r=1}^R \beta_r \SRT(J_r).
\end{equation}
A successful comparison between \Eqref{eq:lhsend} and \Eqref{eq:rhsend} that proves \Eqref{eq:ineq} would correspond to establishing that
\begin{equation}
\label{eq:comp}
    \sum_{x,x'} d_\alpha(x,x') \mathcal{A}\left[\gamma(x,x')\right] \ge \sum_{x,x'} d_\beta(f(x),f(x')) \mathcal{A}\left[\gamma(x,x')\right].
\end{equation}
An obvious sufficient condition\footnote{At face value this condition seems too strong to be necessary. Although this intuition turns out to be correct for \Eqref{eq:cont1} as written, the question becomes more subtle when $\beta$ is made into an all-$1$ vector by simply expanding RHS terms in \Eqref{eq:ineq} into multiple copies of unit coefficient. Upon this innocuous-looking manipulation, it is believed that \Eqref{eq:cont1} becomes a necessary condition -- see~\cite{Avis:2021xnz} for more details on this.} for this to hold is that $f$ be a \textit{contraction map} for the distance functions $d_\alpha$ and $d_\beta$, i.e., that for every $x,x'\in\{0,1\}^L$,
\begin{equation}
\label{eq:cont1}
    d_\alpha(x,x') \ge d_\beta(f(x),f(x')).
\end{equation}
This whole discussion gives rise to the desired proof-by-contraction method, which can finally be stated as follows:
\begin{nthm}[Proof-by-contraction]
\label{thm:pbc}
    An inequality of the form of \Eqref{eq:ineq} holds for the RT formula if there exists a contraction map $f:\{0,1\}^L\to\{0,1\}^R$ for the weighted Hamming distances $d_\alpha$ and $d_\beta$ obeying the occurrence bitstring conditions in \Eqref{eq:occurcond} for all $i\in[n+1]$.
\end{nthm}

We exemplify this proof method in Tab.~\ref{tab:mmi} by exhibiting a contraction map for the MMI inequality given in \Eqref{eq:mmi}. One can visualize the cut-and-paste strategy that this contraction map encodes as follows (see Fig.~\ref{fig:mmirhs}). Every bitstring in the domain $\{0,1\}^3$ of Tab.~\ref{tab:mmi} labels a distinct (possibly disconnected) region in Fig.~\ref{fig:mmilhs}, as specified by RT surfaces of LHS subsystems through \Eqref{eq:sigmaxs}. The contraction map $f$ is then used to form homology regions for RHS subsystems as in \Eqref{eq:rhshomo}. For instance, we see in Tab.~\ref{tab:mmi} that for $J_1=A$, $f(x)_1=1$ only for $x=110$. This means $\sigma_A = \sigma(110)$, the minimal homology region one can form for $A$ as given by the occurrence bitstring, which results in the non-minimal surface shown on the left of Fig.~\ref{fig:mmirhs}. In contrast, for $J_4=ABC$ the contraction map instructs us to include every single bulk region other than the one homologous to $O$, yielding now the minimal surface shown on the right of Fig.~\ref{fig:mmirhs}. One can easily check that no other choice of images for $f$ would obey the contraction property.

\begin{table}[ht!]
\setlength{\tabcolsep}{.3cm}
	\centering
	\begin{tabular}{c || c | c | c||c | c | c | c}
		& AB & AC & BC & ~A~ & ~B~ & ~C~ & ABC \\ \hline
		O & 0 & 0 & 0 & 0 & 0 & 0 & 0 \\
		 ~  & 0 & 0 & 1 & 0 & 0 & 0 & 1 \\
		 ~  & 0 & 1 & 0 & 0 & 0 & 0 & 1 \\
		C & 0 & 1 & 1 & 0 & 0 & 1 & 1 \\
		 ~  & 1 & 0 & 0 & 0 & 0 & 0 & 1 \\
		B & 1 & 0 & 1 & 0 & 1 & 0 & 1 \\
		A & 1 & 1 & 0 & 1 & 0 & 0 & 1 \\
		 ~  & 1 & 1 & 1 & 0 & 0 & 0 & 1 \\
	\end{tabular}
	\caption{
	Tabular representation of the (unique) contraction map that proves validity of the MMI inequality \Eqref{eq:mmi} in holography.
	Occurrence bitstrings as defined in \Eqref{eq:occur} and their images, fixed by \Eqref{eq:occurcond}, are indexed by their associated party label on the left-most column, including the one for the purifier O. Columns are indexed by the bitstring entry they label, with a vertical double-line separating domain from codomain of the contraction map. For the domain, each $I_l$ labels entry $x_l$ for $l\in[L]$ of $x\in \{0,1\}^L$ and, for the codomain, each ${J_r}$ labels entry $y_r$ for $r\in[R]$ of $y\in \{0,1\}^R$. Every row represents one entry of the map $f : x\mapsto y$ by listing input bits in $x$ followed by output bits in $y$.}
	\label{tab:mmi}
\end{table}

\begin{figure}
    \centering
    \includegraphics[width=.95\textwidth]{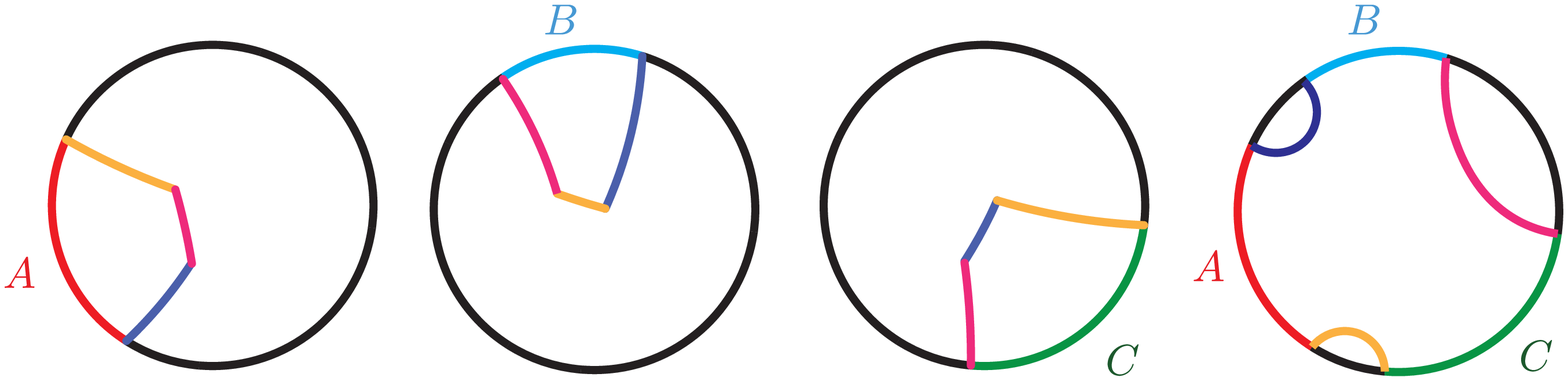}
    \caption{Homology regions specified by the unique contraction map for MMI in Tab.~\ref{tab:mmi}, exemplified by the configuration in Fig.~\ref{fig:mmilhs}. From left to right, subfigures show the RHS homology region and bounding surfaces that the cut-and-paste procedure yields for subsystems $A$, $B$, $C$ and $ABC$.}
    \label{fig:mmirhs}
\end{figure}

\subsection{Bulk HEC implies Boundary HEC} \label{sec:proofHEC}

The previous subsection has taught us how to prove a holographic entropy inequality when the entropies are given by the RT formula, \Eqref{eq:RT}. We would now like to investigate the case in which the entropies are instead given by the QES formula in \Eqref{eq:QES}. Since the generalized entropy, \Eqref{eq:Sgen}, receives contributions from bulk entropies, we will have to understand how this affects the proof-by-contraction method. Besides using the QES rather than RT surfaces, the partitioning of the bulk into regions via inclusion/exclusion carries through identically. 

The first departure we observe involves \Eqref{eq:sibulkxx}, where we now have to include bulk entropy contributions as well. Making use of \Eqref{eq:lhshomo}, the QES entropy can be written in terms of bitstrings as
\begin{equation}
    \SQES(I_l) = \frac{1}{4G}\sum_{x,x'} \abs{x_l - x_l'} \mathcal{A}\left[\gamma(x,x')\right] + \Sbulk\left({\bigcup_{x:x_l=1} \sigma(x)}\right).
\end{equation}
Let $f$ now be a map specifying homology regions for RHS subsystems as originally introduced above \Eqref{eq:rhshomo}. Then, similar to \Eqref{eq:sjbulkxx}, the minimality condition in the definition of the QES entropy guarantees that
\begin{equation}
    \SQES(J_r) \le \frac{1}{4G}\sum_{x,x'} \abs{f(x)_r - f(x)_r'} \mathcal{A}\left[\gamma(x,x')\right] + \Sbulk\left({\bigcup_{x:f(x)_r=1} \sigma(x)}\right).
\end{equation}
Following steps analogous to those below \Eqref{eq:sjbulkxx}, we find that the desired inequality in the form of \Eqref{eq:ineq} that we wish to prove becomes (cf. \Eqref{eq:comp})
\begin{equation}
\label{eq:qescomp}
\begin{aligned}
    \sum_{l=1}^L \alpha_l \Sbulk\left({\bigcup_{x:x_l=1} \sigma(x)}\right) &- \sum_{r=1}^R \beta_r \Sbulk\left({\bigcup_{x:f(x)_r=1} \sigma(x)}\right) \\
    &\ge - \sum_{x,x'} \left( d_\alpha(x,x') - d_\beta(f(x),f(x')) \right) \mathcal{A}\left[\gamma(x,x')\right].
\end{aligned}
\end{equation}

Now, suppose that the inequality in \Eqref{eq:ineq} is a valid HEC inequality for RT and thus, $f$ can be chosen to be a contraction map consistent with Thm.~\ref{thm:pbc}. Upon choosing such an $f$, the right-hand side of \Eqref{eq:qescomp} is guaranteed to be non-positive and the inequality can be collapsed down to a purely bulk entropy inequality of the form
\begin{equation}
\label{eq:bulkineq}
    \sum_{l=1}^L \alpha_l \Sbulk\left({\bigcup_{x:x_l=1} \sigma(x)}\right) \ge \sum_{r=1}^R \beta_r \Sbulk\left({\bigcup_{x:f(x)_r=1} \sigma(x)}\right).
\end{equation}
If this bulk inequality holds, then we would be guaranteed that its boundary counterpart, \Eqref{eq:ineq} with entropies computed by the QES, would hold as well. But \Eqref{eq:bulkineq} looks like a very complicated entropy inequality: since the unions run over $x\in\{0,1\}^L$, the number of distinct regions $\sigma(x)$ now playing the role of parties will generally be as large as $2^L-1$\footnote{It is not $2^L$ because the purifier $O$ always has an all-zero occurrence vector on both sides, so its associated $\sigma(x)$ subregion does not appear anywhere.}. For instance, when \Eqref{eq:ineq} is the MMI inequality in \Eqref{eq:mmi}, the bulk inequality in \Eqref{eq:bulkineq} that the contraction map in Tab.~\ref{tab:mmi} yields is precisely the $7$-party inequality in \Eqref{eq:7mmi}. Hence the relevant question is whether there is any natural condition on the bulk entropy from which inequalities of the form of \Eqref{eq:bulkineq} would follow. The answer is yes, and the condition is that $\Sbulk$ itself obey HEC inequalities as well!

To show this, we will once again employ the proof-by-contraction method. In particular, notice that if $\Sbulk$ obeys HEC inequalities, then it must obey any inequality which can be proved this way. Hence, what we will show is that \Eqref{eq:bulkineq} itself admits a contraction map. Indeed, the contraction map that proves \Eqref{eq:bulkineq} is the very same $f$ that defines it and which recall was assumed to prove the original boundary inequality \Eqref{eq:ineq} for the RT formula. By showing that $f$ is a contraction map for \Eqref{eq:bulkineq}, we will have shown that the boundary entropy $\SQES$ obeys HEC inequalities if the bulk entropy $\Sbulk$ does so too.

The proof is straightforward and just requires thinking about the occurrence bitstrings of \Eqref{eq:bulkineq}. The bulk regions $\sigma(x)$ for every $x\in\{0,1\}^L$ that is not the all-zero bitstring are the $2^L-1$ parties that make up this inequality. Now, the occurrence bitstring of a given party $\sigma(x)$ is nothing but $x$ itself -- this is simply because $\sigma(x)$ shows up for every $l$ for which $x_l=1$, and does not show up otherwise. So every single domain bitstring $x\in\{0,1\}^L$ is itself an occurrence bitstring, including the all-zero bitstring for the purifying region. Suppose we now want to build homology regions for the RHS subsystems via some map $\tilde{f}$. The homology constraint prescribed by \Eqref{eq:occurcond} fixes the image of all occurrence bitstrings to be precisely $\tilde{f}(x)=f(x)$, as follows from \Eqref{eq:bulkineq}. However, this completely fixes the LHS-to-RHS map to be $\tilde{f}=f$, leaving no residual freedom. Since $f$ was born as a contraction map for \Eqref{eq:ineq}, $f$ itself provides the proof-by-contraction that proves \Eqref{eq:bulkineq}.

In practice, looking at MMI as an example, all we really had to do was to take Tab.~\ref{tab:mmi} and assign a label to every unlabelled row. Understanding these as parties, and the LHS and RHS bitstrings as their occurrence bitstrings and respective images then tells us which parties to append to each column label. If we were to label rows top-to-bottom by $O$-$F$-$E$-$C$-$D$-$B$-$A$-$G$, the resulting column labels are precisely the ones which would correspond to \Eqref{eq:7mmi}. With this relabelling, Tab.~\ref{tab:mmi} itself proves \Eqref{eq:7mmi} to be a valid HEC inequality.

Altogether, we have proven the desired result:
\begin{nthm}
\label{eq:hectohec}
    If $\Sbulk$ obeys all HEC inequalities, then $\SQES$ obeys all HEC inequalities too.
\end{nthm}
This is a very general result which extends the relevance of the HEC from RT to the QES prescription.\footnote{A subset of this result, i.e. bulk HEC implies boundary MMI, has been explored in Ref.~\cite{Agon:2021tia} using the formalism of bit threads.} However, it should be clear from our discussion that Thm.~\ref{eq:hectohec} gives a sufficient but highly unnecessary condition on $\Sbulk$ for $\SQES$ to obey any particular HEC inequality. In a more controlled analysis, we could ask: what natural condition on $\Sbulk$ guarantees that $\SQES$ obey a certain HEC inequality?
This is precisely what we turn to next for the case of MMI.

\section{Bulk MMI implies Boundary MMI}
\label{sec:MMI}

The results in this section are simple enough to state and prove and thus, we simply present them without elaborating on the techniques used to obtain them.
We encourage the interested reader to look at App. \ref{asec:MMI} for more details.

For subsets $\varnothing\ne I,J,K\subseteq[n+1]$, the following bits of notation will be useful:
\begin{subequations}
\begin{align}
    \label{eq:saq}
    I_2(I\!:\!J) &\equiv S(I)+S(J) - S(I\!\cup\!J),\\
    \label{eq:ssaq}
    Q(I ; J\!:\!K) &\equiv S(I\!\cup\!J) + S(I\!\cup\!K) - S(I) - S(I\!\cup\!J\!\cup\!K),\\
    \label{eq:mmiq}
    I_3(I\!:\!J\!:\!K)&\equiv S(I\!\cup\!J) + S(I\!\cup\!K) + S(J\!\cup\!K) - S(I)-S(J)-S(K)-S(I\!\cup\!J\!\cup\!K).
\end{align}
\end{subequations}
These may be recognized as defining arbitrary instances of the following inequalities:\footnote{Restrictions on the choices of subsystems $\varnothing\ne I,J,K\subseteq[n+1]$ are needed to prevent these from trivializing -- see Eqs. \eqref{eq:SAins}, \eqref{eq:SSAins} and \eqref{eq:MMIins} for more details on this.}
\begin{subequations}\label{eq:allineqs}
\begin{alignat}{4}
    \label{eq:sadef}
    \text{SA:}  && \qquad I_2(I\!:\!J)      &&\;\ge\;&&0,\\
    \label{eq:ssadef}
    \text{SSA:} && \qquad Q(I ; J\!:\!K)        &&\;\ge\;&&0,\\
    \label{eq:mmidef}
    \text{MMI:} && \quad -I_3(I\!:\!J\!:\!K)&&\;\ge\;&&0.
\end{alignat}
\end{subequations}
The first two, subadditivity (SA) and strong subadditivity (SSA), are universal quantum entropy inequalities (i.e., valid for all quantum states), whereas the third one, monogamy of mutual information (MMI), is a holographic one (i.e., valid only when the RT formula applies). In what follows it is implicitly assumed that $\Sbulk$ obeys SA and SSA.

The first result is a negative but expected one (see e.g. Ref.~\cite{Akers:2019lzs}): that $\Sbulk$ obey the universal inequalities of SA and SSA is not enough for $\SQES$ to obey MMI. In other words,

\begin{nprop}
\label{prop:ghz}
    Boundary MMI does not hold in general.
\end{nprop}
\begin{proof}
    Consider three sufficiently small and distant boundary subregions $A$, $B$ and $C$ such that their geometric contributions to $\SQES$ for any subsystem $I$ built out of them factorizes and takes the form $\SQES(I)= \sum_{i\in I} S_{\mathcal{A}}(i) + \Sbulk(I)$. Evaluating \Eqref{eq:mmiq} on such a configuration, all $S_{\mathcal{A}}$ contributions cancel out and only $\Sbulk$ ones remain. Choosing the bulk state among these subregions and their complement to be of $4$-party GHZ-type, one achieves $\Sbulk(I)=S_0$ for all $I$ subsystems. The upshot is a configuration where the bulk state clearly obeys SA and SSA, but for which $\SQES$ yields $\IN{3}{A:B:C}=2S_0\ge0$, thus violating boundary MMI.
\end{proof}
We can also understand this argument in terms of a violation of the $7$-party condition in \Eqref{eq:7mmi}. Namely, the configuration used in the proof above would correspond to only having the occurrence bitstrings $000$, $110$, $101$ and $011$ labelling nonempty regions in Fig.~\ref{fig:mmilhs}. In terms of the $7$-party inequality derived from \Eqref{eq:bulkineq} using Tab.~\ref{tab:mmi}, this results in \Eqref{eq:7mmi} with $D$, $E$, $F$ and $G$ being empty regions. This reduces \Eqref{eq:7mmi} to the form of MMI for $\Sbulk$, which is clearly violated by a $4$-party GHZ state.

In searching for sufficient conditions on $\Sbulk$ for $\SQES$ to obey MMI, we are thus led to impose MMI in the bulk as well. A natural way to proceed is thus to combine all instances of SA, SSA and MMI in \Eqref{eq:allineqs} for $7$ parties and check if these imply \Eqref{eq:7mmi} as a bulk inequality. This leads us to {\bf Result 2}:
\begin{nprop}
\label{prop:mmitommi}
    Bulk MMI implies boundary MMI
\end{nprop}
\begin{proof}
    The 7-party bulk inequality \Eqref{eq:7mmi} which implies boundary MMI is \cite{Akers:2019lzs}
    \begin{equation}
    \S{ABDG}+\S{ACEG}+\S{BCFG}\ge \S{A}+\S{B}+\S{C}+\S{ABCDEFG}.
    \end{equation}
    Consider the special case that $D = E = F = \varnothing$. What remains is
    \begin{equation}\label{eq:proof_MMI_addG}
        \S{ABG} + \S{ACG} + \S{BCG}  \ge \S{A} + \S{B} + \S{C} + \S{ABCG},
    \end{equation}
    which follows from combining the two inequalities
    \begin{align}
        \S{AG} + \S{BCG} &\ge \S{A} + \S{BC}, \label{eq:MMIproof_weak_monotonicity} \\
        \S{ABG} + \S{ACG} + \S{BC} &\ge \S{AG} + \S{B} + \S{C} + \S{ABCG},
    \end{align}
    where the first is weak-monotonicity for the regions $A, BC, G$ and the second is MMI for the regions $AG, B, C$. (Note that weak-monotonicity can be written as SSA by including the purifier. Let $DEFO$ purify $ABCG$. Then \Eqref{eq:MMIproof_weak_monotonicity} is SSA for the regions $A, DEFO, G$.)
    Now let $D$ be arbitrary. What remains is
    \begin{equation}\label{eq:proof_MMI_addD}
        \S{ABDG} + \S{ACG} + \S{BCG}  \ge \S{A} + \S{B} + \S{C} + \S{ABCDG},
    \end{equation}
    which follows from \Eqref{eq:proof_MMI_addG} combined with
    \begin{equation}
        \S{ABDG} + \S{ABCG}  \ge \S{ABG} + \S{ABCDG}.
    \end{equation}
    This holds from SSA on the regions $ABG, C, D$.
    Similarly, now let $E$ be arbitrary to obtain
    \begin{equation}\label{eq:proof_MMI_addE}
        \S{ABDG} + \S{ACEG} + \S{BCG}  \ge \S{A} + \S{B} + \S{C} + \S{ABCDEG},
    \end{equation}
    which follows from \Eqref{eq:proof_MMI_addD} combined with
    \begin{equation}
        \S{ACEG} + \S{ABCDG}  \ge \S{ACG} + \S{ABCDEG}.
    \end{equation}
    This follows from SSA on the regions $ACG, BD, E$.
    Adding back $F$ works the same way. It returns us to the 7-party inequality \Eqref{eq:7mmi}, which follows from \Eqref{eq:proof_MMI_addE} combined with
    \begin{equation}
        \S{BCFG} + \S{ABCDEG} \ge \S{BCG} + \S{ABCDEFG}.
    \end{equation}
    This follows from SSA on the regions $BCG, ADE, F$.
    This completes the proof, which can be compactly summarized in the following expression:
    \begin{equation}
    \label{eq:mmiproof}
        -\IN{3}{AG:B:C}+\Q{A;DEFO:G}+\Q{ABG;C:D}+\Q{ACG;BD:E}+\Q{BCG;ADE:F}\geq0.
    \end{equation}
\end{proof}

The finding of such an expression is nontrivial but can also be easily tackled as a linear programming problem, as explained in App. \ref{asec:MMI}. We emphasize that there is no obvious reason why Prop.~\ref{prop:mmitommi} should hold: that the same contraction map proves both \Eqref{eq:mmi} and \Eqref{eq:7mmi} (as in the proof of Thm.~\ref{eq:hectohec}) is a priori logically unrelated to the fact that the former implies the latter. Indeed, this is particularly remarkable because only MMI, which effectively is just a $3$-party inequality, suffices to prove validity of \Eqref{eq:7mmi}, which is a $7$-party inequality.

\section{Discussion} \label{sec:disc}

\subsection*{Summary}

So far, the exploration of holographic entropy inequalities had focused on using the RT formula and its geometric nature.
However, we have learned in the past few years that quantum effects incorporated using the QES formula can play a crucial role, e.g., in the presence of black holes.
Thus, the goal of this work is to open up a research program of understanding the connection between bulk and boundary entropy inequalities.

To summarize, we have derived relations between constraints imposed on bulk entropies and the corresponding constraints on boundary entropies by relating them via the QES formula.
The first result we showed in this direction is that the bulk HEC imposes the boundary HEC in a nontrivial way.
A generic $n$-party boundary HEC inequality requires a bulk inequality that could contain an exponentially large number of parties.
Nevertheless, the entire collection of bulk HEC inequalities guarantees the full set of HEC inequalities on the boundary.
Secondly, we showed that imposing MMI on arbitrary subregions in the bulk leads to arbitrary boundary subregions satisfying MMI as well. This result is a first step in carrying out a more controlled study of the interplay between specific bulk and boundary constraints.

\subsection*{Assumption}
Note that, in stating both of our main results, we assumed the QES prescription is valid. This is typically expected to be true when it is applied in semiclassical bulk states, i.e., states of quantum fields on fixed, possibly curved backgrounds. However, it has recently been pointed out that there are certain semiclassical states for which the QES prescription gives the wrong entropy at leading order in $1/G$~\cite{Akers:2020pmf}. One could worry that in such states, our results will not hold, for example that bulk MMI will no longer imply boundary MMI.
We expect that this is actually not a problem. That is, even with the leading order corrections of Ref.~\cite{Akers:2020pmf}, we expect that our results continue to be valid. The basic idea is that different parts of the bulk wavefunction will each satisfy the QES prescription, obeying our results, and the inequalities in the HEC are such that their mixture will therefore also obey our results. We leave a detailed analysis of this for future work.

\subsection*{Conjectures}

Our results here can be interpreted as the two extremities of a general set of such connections between boundary and bulk entropies.
We have analyzed the HEC-constrained QES cone, which we showed to be equal to the HEC, and also the MMI-constrained QES cone, which could be at most as large as the MMI cone.
There are various ways in which one could imagine interpolating between our two results.
We conjecture a list of possibilities here:
\begin{itemize}
    \item If a given $n$-party inequality no stronger than HEC ones is satisfied in the bulk, then the same $n$-party inequality is satisfied on the boundary.
    \item If the complete set of $n$-party HEC inequalities is satisfied in the bulk, then every $n$-party HEC inequality is satisfied on the boundary.
    \item If the complete set of $k$-party HEC inequalities is satisfied in the bulk, then there is some $n_k$ such that every $n$-party HEC inequality is satisfied on the boundary for all $k\le n \le n_k$.
\end{itemize}
In the first possibility above we needed to exclude inequalities stronger than HEC ones because those cannot possibly by satisfied by any constrained QES cones since, as we mentioned before, all such cones automatically contain the HEC due to the area term.
On the other hand, notice that we did not restrict this possibility to just HEC inequalities -- indeed, it is possible that, in full generality, any weaker bulk constraint implies its boundary counterpart. Such a situation may be hard to prove but easy to falsify.

When focusing on HEC inequalities, the three possibilities above can be easily seen to go from strong to weak, in the sense that each of them would imply the subsequent one. For instance, the second one implies the trivial case $n_k=k$ for the third one, which in principle could be weaker by having $n_k\ge k$.
It would be interesting to explore these and other logical possibilities in future work.

\subsection*{Hypergraphs}

Holographic states can be understood as states whose entanglement structure admits a suitably discretized representation in terms of graphs where the RT formula computes the entropy of a boundary subregion in terms of the minimum cut across the graph.
A simple generalization of this class of states is to hypergraphs, where nodes can be connected using hyperedges instead of regular edges.
A $k$-hyperedge is a connection that groups $k$ nodes simultaneously, where $k=2$ corresponds to usual edges.
These states also satisfy a nontrivial entropy cone if one posits that the entropy of a boundary subregion is computed using a generalized RT formula, i.e., using the minimum cut in a hypergraph~\cite{Bao:2020zgx,Bao:2020mqq}.
Further, such hypergraph states can be explicitly constructed as stabilizer states using random tensor networks with $k$-party GHZ links instead of Bell-pair-like bonds~\cite{Walter:2020zvt}.

There is also a simple generalization of the QES formula to hypergraphs, taking a form analogous to \Eqref{eq:QES}.
In the hypergraph, the homology region $\sigma_{R}$ is described by the collection of nodes defining the minimum cut for a set of boundary nodes $R$. The area term in the generalized entropy is implemented by the total weight of the minimum cut, which is given by the sum of the weights of all hyperedges that connect nodes in $\sigma_R$ to those in its complement. Finally, the bulk entropy contribution can also be realized in a random tensor network in the form of bulk dangling legs on every node, analogous to Ref.~\cite{Hayden:2016cfa}.

Proving entropy inequalities for the RT formula on hypergraph states involves a generalization of the contraction-map technique used in Sec.~\ref{sec:contraction}.
Roughly, apart from requiring that the map $f$ described in Sec.~\ref{sec:contreview} obey a contraction property for pairs of bitstrings, the proof of inequalities on hypergraphs involves additional multi-bitstring contraction conditions  (see Ref.~\cite{Bao:2020zgx} for more details). Intuitively, these extra conditions make it strictly harder for such an $f$ to exist, thereby explaining why hypergraph entropies obey weaker inequalities than graph ones and thus attain richer entanglement structures.
Nevertheless, we can again run a similar argument to that in Sec.~\ref{sec:proofHEC} to show in an analogous fashion that the bulk inequalities required to prove any specific hypergraph cone boundary inequality follow from the same contraction map upon relabelling parties.
Thus, one can again see that the bulk hypergraph cone implies the boundary hypergraph cone.

It would also be interesting to probe the relation between specific hypergraph inequalities in the bulk and boundary, in the spirit of Sec.~\ref{sec:MMI}. For instance, in Ref.~\cite{Bao:2020zgx} it was shown that minimum cuts on hypergraphs obey the Ingleton inequality~\cite{ingleton1971representation}
\begin{align}
    \S{AB}\!+\!\S{AC}\!+\!\S{AD}\!+\!\S{BC}\!+\!\S{BD} \ge \S{A}\!+\!\S{B}\!+\!\S{CD}\!+\!\S{ABC}\!+\!\S{ABD}.
\end{align}
However, since this inequality involves $L=5$ terms, its hypergraph contraction map $f$ would lead to a $2^L-1=31$-party bulk inequality, so it seems rather nontrivial to prove a result analogous to the one proved for MMI in Sec.~\ref{sec:MMI}.

%\begin{itemize}
%    \item Implications of HEC inequalities on Page curve in doubly holographic models?
%    \item Our result implies the following: all QES in any doubly-holographic experiment (e.g. black hole evaporation, where classical and quantum contributions to entropy are nontrivial) can be equivalently described in a geometry with \emph{no matter fields}. In other words, there exists geometries yielding all former QES entropies for any boundary region as just RT entropies. By 1904.04834, this means that there exists a unique geometry that is in this sense dual to an evaporating black hole. Seems curious!
%    \item Many inequalities for many different $n$ from each contraction proof (exponential in the number of LHS terms in the inequality). What makes them stronger/weaker? What relates them? Space of proofs (Michael Walter)?
%    \item Conjectures for higher parties
%    \item Hypergraphs!
%    \item Emphasize that MMI stuff applies with time-dependence!
%\end{itemize}

%~~~~~~~~~~~~~~~~~~~~~~~~~~~~~~~~~~~~~~~~~~~~~~~~~~~~~~~~~~~~~~~~~~~~~
\acknowledgments
%~~~~~~~~~~~~~~~~~~~~~~~~~~~~~~~~~~~~~~~~~~~~~~~~~~~~~~~~~~~~~~~~~~~~~

We thank Veronika Hubeny and Max Rota for detailed comments on an earlier draft.
We would also like to thank Ning Bao, Netta Engelhardt and Mykhaylo Usatyuk for useful discussions.
CA is supported by the US Department of Energy grants DE-SC0018944 and DE-SC0019127, and also the Simons Foundation as a member of the It from Qubit collaboration.
SHC is supported by NSF grant PHY-2107939, and by a Len DeBenedictis Graduate Fellowship.
At the early stages of this work, SHC was also supported by NSF grant PHY-1801805, and by funds from UCSB.
PR is supported in part by a grant from the Simons Foundation, and by funds from UCSB.
This material is based upon work supported by the Air Force Office of Scientific Research under award number FA9550-19-1-0360.

\appendix

\section{Behind the scenes of Sec. \ref{sec:MMI}}
\label{asec:MMI}

In this section we initiate a more refined study of the relationship between bulk entropy constraints and boundary inequalities for the QES entropy. We begin this program by trying to obtain a weak condition guaranteeing that boundary MMI holds. We know that a sufficient condition for boundary MMI to hold is for the bulk to obey the $7$-party inequality in \Eqref{eq:7mmi}, which we explained can be derived from \Eqref{eq:bulkineq}. We would like to find simpler entropy conditions for smaller party number implying \Eqref{eq:7mmi}.

Our formalization of this search will be as follows.
We are after some basic set of linear inequalities implying another inequality.
A general system of positive linear inequalities can be specified by
\begin{equation}
\label{eq:Ax0}
    Ax\ge0 \qquad\text{where}\qquad A \in \mathbb{R}^{m\times n}, \quad x\in\mathbb{R}_{\ge0}^n,\quad 0\in\mathbb{R}^n.
\end{equation}
The subspace of vectors $x\in\mathbb{R}^n$ allowed by this system of inequalities is known as the feasible region and reads
\begin{equation}
    U = \{x\in\mathbb{R}_{\ge0}^n \;:\; Ax\ge0 \}.
\end{equation}
We say that another inequality $b^Tx\ge0$ specified by some $b\in\mathbb{R}^n$ is redundant with respect to or implied by \Eqref{eq:Ax0} if, when added to this system, the feasible region remains unchanged.
The inequality $b^Tx\ge0$ can be easily seen to be redundant if and only if (see e.g.~\cite{telgen1983identifying})
\begin{equation}
\label{eq:redtest}
    \min\;\{b^Tx \;:\; x\in U\} \ge0,
\end{equation}
a simple diagnostic which can be formulated as a linear programming problem. Notice that because all the inequalities we are considering are homogeneous, redundancy will always yield $0$ in this minimization, corresponding to the all zero-vector $x=0$. In contrast, when $b^Tx\ge0$ is non-redundant the minimization above will be unbounded from below. When the redundancy test in \Eqref{eq:redtest} is obeyed, we will also be interested in obtaining an explicit realization of $b^Tx\ge0$ in terms of a conical combination of the row inequalities in $Ax\ge0$. This problem can be formalized by the statement that $b^Tx\ge0$ is redundant if and only if there exists a conical combination of the rows of $A$ giving $b$, i.e., if there exists a vector $q\in\mathbb{R}_{\ge0}^m$ such that
\begin{equation}
\label{eq:redundant}
    A^T q=b.
\end{equation}
Such a conical combination can again be obtained by solving a linear program as follows:
\begin{equation}
\label{eq:lpsol}
    \text{Find }~ q\in\mathbb{R}^m ~\text{ minimizing }~ c^Tq ~\text{ s.t. }~ A^T q=b ~\text{ and }~ q\ge0,
\end{equation}
where, in order for $q$ to correspond to a conical combination as simple as possible, we will choose the objective function $c^Tq$ to involve an all-$1$ vector $c$. This way, by making $\sum_{i=1}^m q_i$ as small as possible, the output $q$ of \Eqref{eq:lpsol} will correspond to a minimal set of rows of $Ax\ge0$ yielding $b^Tx\ge0$.

Having laid out the basic formalism, we can now try to obtain a natural set of bulk inequalities $Ax\ge0$ implying \Eqref{eq:7mmi} as a redundant one. For any quantum theory, $\Sbulk$ obeys the universal inequalities of subadditivity and strong subadditivity, so these should certainly be part of our system of bulk inequalities. As explained in the introduction, if we are working with a set of parties $[n]$, entropy vectors belong to $\mathbb{R}^{2^n-1}$. A standard way of writing subadditivity (SA) is
\begin{equation}
    \S{A}+\S{B}\ge\S{AB}.
\end{equation}
We would like to apply this inequality to any two non-spanning, disjoint subsets $I,J\subseteq[n+1]$ replacing $A$ and $B$. Using the mutual information, \Eqref{eq:saq}, this results in the following set of all possible instances of SA in $\mathbb{R}^{2^n-1}$:\footnote{Recall that we are interested in mixed states on $[n]$ and that party $n+1$ is a purification. Since we have a pure state on $[n+1]$, entropies enjoy the symmetry $S(I)=S(I^\complement$), where $I^\complement \equiv [n+1] \smallsetminus I$. It is standard to fix this redundancy by excluding the label $n+1$ from entropy expressions -- i.e. if $I\ni n+1$ we write $S(I^\complement)$ instead of $S(I)$. When this happens in our pair of subsystems $I$ and $J$ for SA, we have to replace $I \to I^\complement$ and also $I\cup J \to (I\cup J)^\complement$. The resulting inequality no longer takes the standard form of SA. Indeed, what one obtains is known as the Araki-Lieb inequality, which one can see written as $\S{A}+\S{AB}\ge\S{B}$. The statement is thus that these two types of inequalities are in fact related by the purification symmetry of quantum entropies. Our treatment and notation here unifies such inequalities and we thus refer to them simply as instances of SA.}
\begin{equation}
\label{eq:SAins}
    I_2(I:J) \ge 0, \qquad \forall I,J\subseteq[n+1] ~~\text{ s.t. }~~ 
    \begin{cases}
        I\cap J=\varnothing, \\
        I\cup J\ne [n+1].
    \end{cases}
\end{equation}
The case of $I\cup J = [n+1]$ is excluded because $S([n+1])=0$ by purity, which trivializes the inequality since entropies are non-negative. We can also apply the standard form of SSA in \Eqref{eq:SSA} to any three non-spanning, disjoint subsets $I,J,K\subseteq[n+1]$. Doing so and making use of \Eqref{eq:ssaq}, we obtain all possible instances of SSA in $\mathbb{R}^{2^n-1}$:\footnote{Something analogous to what was explained in the previous footnote for the case of instances of SA happens for SSA as well. If e.g. $J\ni n+1$, one has to replace $I\cup J$ and $I\cup J\cup K$ with their complements, arriving at an inequality known as weak monotonicity which takes the standard form $\S{AB}+\S{BC}\ge\S{A}+\S{C}$. It turns out that if $I\ni n+1$, one in fact gets back an inequality of the standard SSA form.}
\begin{equation}
\label{eq:SSAins}
    Q(I;J:K) \ge 0, \qquad \forall I,J,K\subseteq[n+1] ~~\text{ s.t. }~~
    \begin{cases}
        I\cap J = I \cap K = J \cap K = \varnothing, \\
        I\cup J \cup K \ne [n+1].
    \end{cases}
\end{equation}

Let's begin by including all instances of SA and SSA from Eqs. \eqref{eq:SAins} and \eqref{eq:SSAins} in our set of inequalities obeyed by the bulk entropy. Working with $n=7$ and applying the diagnostic of redundancy in \Eqref{eq:redtest}, to the $7$-party inequality in  \Eqref{eq:7mmi}, we find no redundancy,\footnote{We are claiming here that \Eqref{eq:7mmi} is non-redundant with respect to the system of inequalities consisting of all instances of SA and SSA. However, one does have plenty of redundancy within the system of inequalities itself~\cite{pippenger2003inequalities}. This does not affect the argument, since the feasible region remains the same whether or not one includes redundant instances of SA and SSA.} thus obtaining our first result:
\begin{nfact}
    Bulk SA and SSA do not imply boundary MMI.
\end{nfact}
This was of course expected from the simple argument in Prop.~\ref{prop:ghz}, but a good consistency check nonetheless. Let us consider adding more stringent inequalities to our system of bulk inequalities. A natural one from the context of holography is MMI itself, which we would now like to impose on the bulk entropy. We can apply the standard form of MMI from \Eqref{eq:mmi} to any triple of non-spanning, disjoint subsets $I,J,K\subseteq [n]$. Using the tripartite information, \Eqref{eq:mmiq}, we obtain this way all possible instances of MMI in $\mathbb{R}^{2^n-1}$:\footnote{That we are now taking only subsets of $[n]$ rather than of $[n+1]$ is not a typo. In fact, this is because including e.g. $I\ni n+1$ would give back an instance of MMI where $I\to (I\cup J \cup K)^\complement \not\ni n+1$, with $J$ and $K$ unchanged. This follows from a nontrivial symmetry property that all genuinely holographic entropy inequalities happen to share known as superbalanced (see~\cite{Hubeny:2018ijt,He:2020xuo,Avis:2021xnz} for more details).}
\begin{equation}
\label{eq:MMIins}
    - I_3(I:J:K) \ge 0, \qquad \forall I,J,K\subseteq[n+1] ~~\text{ s.t. }~~
    \begin{cases}
        I\cap J = I \cap K = J \cap K = \varnothing, \\
        I\cup J \cup K \ne [n+1].
    \end{cases}
\end{equation}
We can now go ahead and add all instances of MMI to our system of entropy inequalities involving SA and SSA. However, observe that the addition of MMI turns all instances of SSA into redundant inequalities, since any of the latter becomes a conical combination of SA and MMI via the identity
\begin{equation}
\label{eq:ssared}
    Q(I;J:K) = - I_3(I:J:K) + I_2(J:K).
\end{equation}
As a first step, we may thus brute force all possible instances of MMI into our set of inequalities and remove SSA from it completely. It also turns out that all instances of SA except for those involving singletons when written as in \Eqref{eq:SAins} are redundant in this system. Equivalently, only the MMI inequalities together with the singleton instances of SA are associated to facets of the polyhedral cone defined by all instances of SA and MMI.
Hence it is natural to only employ these in the redundancy test of \Eqref{eq:redtest}, with which we arrive at:
\begin{nfact}
    Bulk SA and MMI imply boundary MMI.
\end{nfact}
Given that we found redundancy, we may look more closely at the linear program in \Eqref{eq:lpsol} in order to obtain a solution to it. The simplest expressions one can obtain for the $7$-party inequality in \Eqref{eq:7mmi} in terms of a non-redundant system of instances of SA and MMI happens to be quite complicated in fact. Involving no fewer than $21$ terms, an example of these is
\begin{equation}
\begin{aligned}
&
-\IN{3}{A:E:G}
-\IN{3}{A:F:CDG}
-\IN{3}{B:E:CF}
-\IN{3}{B:ADG:CEF}
-\IN{3}{C:D:F} \\
&
-\IN{3}{C:D:AG}
-\IN{3}{D:E:ACFG}
-\IN{3}{D:G:BCF}
-\IN{3}{E:F:ACG}
-\IN{3}{F:G:CD} \\
&
-\IN{3}{G:AE:BCDF}
+\IN{2}{A:O}
+\IN{2}{F:G}
+\IN{2}{E:G}
+\IN{2}{E:F}
+\IN{2}{D:G} \\
&
+\IN{2}{D:F}
+\IN{2}{D:E}
+\IN{2}{C:D}
+\IN{2}{B:E}
+\IN{2}{A:F}
\ge 0,
\end{aligned}
\end{equation}
where expanding out the left-hand side in terms of entropies leads to \Eqref{eq:7mmi}. There is a lesson to be drawn from how complicated this is: heuristically, the fact that so many terms are needed to yield \Eqref{eq:7mmi} means that this inequality is highly redundant with respect to the set of all non-redundant instances of SA and MMI. This suggests that we may also accomplish redundancy and a simpler expression by restoring the redundant instances of SA and the weaker SSA inequalities that we removed. This indeed turns out to be a successful strategy which allows one to show that, in fact, one only needs to use a single instance of MMI to render \Eqref{eq:7mmi} redundant. Remarkably, the minimal representations of \Eqref{eq:7mmi} that one obtains now involve just $5$ terms, and an example different from \Eqref{eq:mmiproof} is
\begin{equation}
\label{eq:simple7}
%ABCG:ABD
%ABCDFG:ACEG
%ABCDG:BCFG
%CG:CEFO OR ABDG:ABDEFO
    -\IN{3}{A:B:CG}+\Q{AB;CG:D}+\Q{ACG;BDF:E}+\Q{BCG;AD:F}+\Q{C;EFO:G}\geq0.
\end{equation}
Notice that this expression no longer involves instances of SA, which means it is not needed to imply \Eqref{eq:7mmi} in the presence of SSA and MMI. Hence,
\begin{nfact}
    Bulk SSA and MMI imply boundary MMI.
\end{nfact}
It is worth emphasizing that SSA and MMI are effectively $3$-party inequalities, which we just lifted to $n=7$ by applying them to parties grouped into larger subsystems. In this sense, we have found that our genuine $7$-party inequality in \Eqref{eq:7mmi} is in fact implied by much simpler $n=3$ inequalities. Interestingly, the highest lift of MMI that turns out to be needed involves $\abs{I\cup J \cup K}=4$ parties, rather than $7$ ($3$ is not enough). On the other hand, one does need to utilize at least one instance of SSA involving $\abs{I\cup J\cup K}=7$ parties, as can be seen in \Eqref{eq:simple7}.

%~~~~~~~~~~~~~~~~~~~~~~~~~~~~~~~~~~~~~~~~~~~~~~~~~~~~~~~~~~~~~~~~~~~~~
\addcontentsline{toc}{section}{References}
\bibliographystyle{JHEP}
\bibliography{references}
%~~~~~~~~~~~~~~~~~~~~~~~~~~~~~~~~~~~~~~~~~~~~~~~~~~~~~~~~~~~~~~~~~~~~~

\end{document}